\documentclass{toc}
\newcommand{\eat}[1]{}
\newtheorem{observation}[theorem]{Observation}

\newlength{\myindentrightmargin}
\newlength{\myindent}
\newlength{\myindentdelta}
\newlength{\mytextwidth}
\newcommand{\END}{%
  \addtolength{\myindent}{-\myindentdelta}
  \addtolength{\mytextwidth}{+\myindentdelta}
}

\newcounter{myenum}             %pg 113
\newenvironment{myenum}{
  \setlength{\myindent}{0pt}
  \setlength{\parskip}{0ex}
  \begin{list}
    {\small\arabic{myenum}.}{\usecounter{myenum}
      \setlength{\topsep}{0ex}
      \setlength{\parskip}{0ex}
      \setlength{\partopsep}{0ex}
      \setlength{\leftmargin}{1em}
      \setlength{\rightmargin}{0em}
      \setlength{\itemsep}{1ex}
      \setlength{\parsep}{0ex}
      \setlength{\listparindent}{0ex}
      \setlength{\labelwidth}{0.65em}
      \setlength{\labelsep}{0.2em}
      \setlength{\itemindent}{0em}
    }
  }{\end{list}}

\renewcommand{\N}{{\ensuremath{\cal N}}}
\newcommand{\U}{{\ensuremath{\cal U}}}
\newcommand{\W}{{\ensuremath{\cal R}_\U}}
\renewcommand{\R}{{\ensuremath{\cal R}}}

\renewcommand{\P}{\ensuremath{( n^{4\alpha} \sqrt{\alpha \ln n})}}
\newcommand{\Plog}{\ensuremath{( n^{4\alpha} \sqrt{\alpha \log n})}}

\newcommand{\PP}{\ensuremath{( n^{4\alpha+1}\sqrt{\alpha \ln n})}}
\newcommand{\PPlog}{\ensuremath{( n^{4\alpha+1}\sqrt{\alpha \log n})}}

\renewcommand{\H}{\ensuremath{\sqrt{\alpha n \ln n}}}
\newcommand{\Hlog}{\ensuremath{\sqrt{\alpha n \log n}}}

\newcommand{\cclass}[1]{\textsf{#1}}

\tocpdftitle{ Hamming Approximation of NP Witnesses }

\copyrightauthor{Daniel Sheldon and Neal E. Young}
\tocpdfauthor{Daniel Sheldon and Neal E. Young}

\begin{document}
\begin{frontmatter}
\title{ Hamming Approximation of NP Witnesses}
\author[sheldon]{Daniel Sheldon\thanks{Supported by NSF Award IIS-1125098.}}
\author[young]{Neal E. Young\thanks{Supported by NSF Award CCF-1117954.}}
\begin{abstract}
  {\em Given a satisfiable 3-SAT formula,
  how hard is it to find an assignment to the variables 
  that has Hamming distance at most $n/2$ to a satisfying assignment?}
  More generally, consider any polynomial-time verifier for any \cclass{NP}{}-complete language.
  A {\em $d(n)$-Hamming-approximation algorithm} for the verifier 
  is one that, given any member $x$ of the language, 
  outputs in polynomial time a string $a$ with Hamming distance 
  at most $d(n)$ to some witness $w$, where $(x,w)$ is accepted by the verifier.
  Previous results have shown that, if \cclass{P}{}$\ne$\cclass{NP}{}, every \cclass{NP}{}-complete language has a verifier
  for which there is no $(n/2-n^{2/3+\delta})$-Hamming-approximation algorithm,
  for various constants $\delta\ge 0$.

  Our main result is that, if \cclass{P}{}$\ne$\cclass{NP}{}, then every paddable \cclass{NP}{}-complete language
  has a verifier that admits no $(n/2+O(\sqrt{n\log n}))$-Hamming-approximation
  algorithm.
  That is, one can't get even {\em half} the bits right.
  We also consider {\em natural} verifiers for various well-known \cclass{NP}{}-complete problems.
  They do have $n/2$-Hamming-approximation algorithms,
  but, if \cclass{P}{}$\ne$\cclass{NP}{}, have no $(n/2-n^\epsilon)$-Hamming-approximation algorithms
  for any constant $\epsilon>0$.

  We show similar results for randomized algorithms.
\end{abstract}
\tocacm{F.1.3}
\tocams{68Q17}
\tockeywords{complexity theory, inapproximability, approximation algorithms, Hamming distance}
\end{frontmatter}

%%%%%%%%%%%%%%%%%%%%%%%%%%%%%%%%%%%%%%%%%%%%%%%%%%%%%%%%%%%%
% Introduction
%%%%%%%%%%%%%%%%%%%%%%%%%%%%%%%%%%%%%%%%%%%%%%%%%%%%%%%%%%%%

\newpage
\section{Introduction}

Consider the discrete tomography problem.  An instance is specified by
numerous two-dimensional x-ray images, each formed by x-raying a
three-dimensional object at some angle.  A solution is
a description of the internal structure of the object,
specifying the density of matter (0 or 1) within each voxel in the object.  
Given sufficiently many x-ray images, the internal structure can be
determined uniquely, yet computing it exactly is (in general)
\cclass{NP}{}-complete.  What form of approximate solution is appropriate?
One standard measure of an approximate solution would be
the extent to which it yields approximately the same x-ray images.
By this measure, a ``good'' approximate solution can have a very
different internal structure than any feasible solution.
If the goal is to discover the internal structure of the object,
then the {\em Hamming distance} to a feasible solution is a more appropriate metric.
This paper is about the computational complexity of computing (possibly infeasible) solutions
that have small Hamming distance to feasible solutions, for various \cclass{NP}{}-complete problems.

\subsection{Previous results}

\paragraph{Definition of Hamming approximation.}
Here we follow~\cite{Feige00hardness,Kumar99proofs} conceptually, but with different terms.
By definition, every \cclass{NP}{} language $L$ has a {\em verifier} $V$, 
such that $L= \{ x : V(x,w) \mbox{ accepts for some } w\}$,
where $V$ runs in time polynomial in $|x|$.
A string $w$ such that $V(x,w)$ accepts is a {\em witness} for $x\in L$.  

A {\em $d(n)$-Hamming-approximation algorithm} for a verifier $V$
(or just a {\em $d(n)$-approximation algorithm}, since this is the only
form of approximation considered here) is an algorithm that, given any
input $(x,n)$, outputs an $n$-bit string $a$ that has Hamming distance 
at most $d(n)$ to some $n$-bit witness $w$ such that $V(x,w)$ accepts,
as long as there is such a witness.
If randomized, the algorithm is a $d(n)$-approximation algorithm for $V$
{\em with probability $p(n)$} provided it outputs such an $a$ with
probability at least $p(n)$.

\paragraph{Arbitrary verifiers are hard to approximate  within $n/2-n^{2/3+\epsilon}$.}
Kumar and Sivakumar were the first to study the hardness of
Hamming approximation~\cite{Kumar99proofs}. 
They showed that if \cclass{P}{}$\ne$\cclass{NP}{}
every \cclass{NP}{}-complete language has {\em some} verifier 
for which there is no $d(n)$-approximation algorithm,
where $d(n)=n/2 - n^{4/5+\delta}$ for some $\delta>0$.
Their verifiers do not use the ``natural'' witness $w$,
but instead use an error-correcting encoding of $w$,
so that even if nearly half the bits of the encoding are corrupted
$w$ can still be recovered.

The hardness threshold $d(n)$ was increased to $n/2-n^{3/4+\epsilon}$ for any $\epsilon > 0$
by Guruswami and Sudan~\cite{guruswami2000list}, 
then increased further to $n/2-n^{2/3+\epsilon}$ for any $\epsilon > 0$
by Guruswami and Rudra~\cite{guruswami2011soft}
via the construction of stronger codes.
Guruswami and Rudra also showed that this approach
cannot prove a hardness threshold above $n/2 - \Theta(\sqrt{n \log n})$.
(Note:~\cite{guruswami2005list,guruswami2011soft}
cite an unpublished draft of the current paper from 2003~\cite{sheldon2003}
that contained the core theorem from the current paper.)
 
\paragraph{Many natural verifiers are hard to approximate within $n/2-n^{1-\delta}$.}
Feige et al.~\cite{Feige00hardness} considered the hardness of approximating 
natural verifiers for specific well-known \cclass{NP}{}-complete problems.
(As an example of a ``natural'' verifier, the natural verifier for 3-SAT 
uses witnesses that are $n$-bit strings, each bit encoding the truth value of one of $n$ variables.)
A motivating application was to give evidence
that a SAT algorithm by Sch{\"o}ning couldn't be sped up in a particular way.
Feige et al.\ leverage~\cite{Kumar99proofs},
using amplification arguments to extend the results to natural verifiers.
Assuming \cclass{P}{}$\ne$\cclass{NP}{},
for some small $\delta > 0$,
for many standard \cclass{NP}{}-complete problems,
they showed that the natural verifier
has no $(n/2-n^{1-\delta})$-approximation algorithm
(even one that works with probability $1/n^c$, assuming \cclass{RP}{}$\ne$\cclass{NP}{}).
This result is weaker than previous results in that the hardness threshold is lower, 
but stronger in that it holds for natural verifiers.

%% Guruswami and Rudra~\cite{guruswami2011soft} 
%% use error-correcting codes to strengthen the hardness results of Feige et al.
%% Specifically, they show
%% that, for every \cclass{NP}{}-complete problem, and every $\epsilon>0$,
%% there is a formulation of the problem
%% for which it is \cclass{NP}{}-hard to achieve Hamming distance less than $n/2 - n^{2/3+\epsilon}$,
%% reducing the exponent in the subtracted term from just less than 1 down to about 2/3.
%See also the thesis of Guruswami~\cite{guruswami2005list} 

\paragraph{Other related work.}
G{\'a}l et al.~\cite{gal1999computing} 
studied self-reductions of \cclass{NP}{}-complete problems to their variants that require
computing a partial witness.
For example, their Thm.~1 says that, given a SAT instance $\psi$ on $n$ variables, 
one can construct another SAT instance $\psi'$ on $N=n^{O(1)}$ variables  
such that (with high probability), given the values of any $N^{1/2 + \epsilon}$ of the $N$ variables 
in a satisfying assignment for $\psi'$, one can
compute in polynomial time a satisfying assignment for $\psi$.
(Such results imply that,
if \cclass{RP}{}$\ne$\cclass{NP}{}, then,
given a SAT instance $\psi$ on $n$ variables,
{\em one cannot compute in polynomial time
a partial assignment, to any $n^{1/2+\epsilon}$ of the variables,
that has an extension to a satisfying assignment.}
However, an easy direct argument shows a stronger result:
assuming \cclass{P}{}$\ne$\cclass{NP}{}, given any SAT instance, {\em one cannot compute in
polynomial time a partial assignment, to even any single variable, 
that has an extension to a satisfying assignment.}  This is simply because, 
by fixing the value of that one variable correctly, simplifying the formula and 
iterating, one could compute a satisfying assignment.)

\medskip
All of the previous works above rely directly or indirectly on error-correcting codes.

\subsection{New results}

\paragraph{The ``universal'' \cclass{NP}{} verifier is hard to approximate within $n/2+O(\sqrt{n\log n})$.}
Our core theorem
(Thm.~\ref{theorem:main})
concerns the hardness of approximating the standard verifier
$V_\U$ for the ``universal'' \cclass{NP}{}-complete language $\U$ (to be defined shortly).
It gives the  first hardness result {\em above} the $n/2$ barrier:
\begin{quote}
  \em
  Fix any constant $\alpha>0$.
  If the verifier $V_\U$ has an $(n/2 + \Hlog)$-Hamming-approximation algorithm $A_\U$, then \cclass{P}{}=\cclass{NP}{}.
  If the verifier has a randomized $(n/2+\Hlog)$-Hamming approximation algorithm
  that works with probability $1-O(1/\PP)$, then \cclass{RP}{}=\cclass{NP}{}.
\end{quote}
The basic idea for the proof is as follows.
Given a potential instance of $\U$, to decide in polynomial time whether the instance is in $\U$,
run $A_\U$ to find a string $a$ such that,
if any witness exists, then some witness must be ``close'' to $a$.
Eliminate from the universe of potential witnesses all strings that are too far from $a$.
This is a polynomial fraction of the universe (Lemma~\ref{lemma:rand}).
Repeat, each time eliminating a polynomial fraction of the remaining candidates.
After polynomially many iterations, the universe of remaining candidates
will have polynomial size.  Use the verifier on each to check if it's a witness.

Of course, calling $A$ with the {\em same} input repeatedly
won't eliminate additional candidates.  Instead we modify the verifier each iteration,
mapping all $u$ remaining candidate strings to a more ``compact'' set:
the $u$ lexicographically smallest bit strings (among those with $\lceil\log_2 u\rceil$ bits),
and recursing on the compacted universe of candidates.

\paragraph{Definition of the universal \cclass{NP}{}-complete language $\U$.}
Before discussing further results, here are the definitions
of the universal \cclass{NP}{}-complete language $\U$ and its natural verifier,
and three  relevant properties.

Language $\U$ contains
the triples $(V,x,1^t)$ where $V$ is (the encoding of) any verifier,
$x$ is a string,
and $1^t$ is a ``padding'' string of $t$ ones,
such that, for some {\em witness} $w$ of length at most $t$,
$V(x,w)$ accepts within $t$ steps.
The verifier $V _\U$ takes as input any pair $((V,x,1^t),w)$
and accepts if $V(x,w)$ accepts within $t$ steps.
The following facts are well known:
\begin{fact}\it
  $\U$ is \cclass{NP}{}-complete and $V_\U$ is a verifier for $\U$.
\end{fact}
It is easy to see that $V_\U$ 
is also as hard to Hamming approximate as any other \cclass{NP}{} verifier:
\begin{fact}\label{fact:hardest}\it
  If there is a $d(n)$-Hamming-approximation algorithm
  (with probability $p(n)$) for $V_\U$,
  then, for any verifier $V$ for any language in \cclass{NP}{},
  there is a $d(n)$-Hamming approximation algorithm
  (with probability $p(n)$) for $V$.
\end{fact}
Fact~\ref{fact:hardest} follows directly from the fact 
that $V$ accepts $(x,w)$ if and only if $V_\U$ accepts $((V, x, 1^t), w)$,
for an appropriately chosen $t$.
So, if a string $a$ approximates a witness $w$ such that $V_\U((V,x,1^t),w)$ accepts,
then $a$ equally well approximates a witness (the same one) such that $V(x,w)$ accepts.

\eat{
\begin{proof}
  Let $M_L$ be a Turing machine that decides the witness relation
  $(x,w)\in \R _L$ in polynomial time $t(|x|)$.
  Let $1^t$ be a string of $t(|x|)$ 1's.

  Given $(x,n)$, the algorithm $A_L$ does the following:
  run $A_\U$ on $(M_L,x,t,n)$ and return the result.
  This takes polynomial time because $A_\U$ takes
  polynomial time and $t = t(|x|)$ is polynomial in $|x|$.

  By the definition of the universal language,
  the size-$n$ witnesses for $x\in L$
  are exactly the size-$n$ witnesses for $(M_L,x,t)\in\U$.
  Because the witness sets are identical, the Hamming-distance
  guarantee of $A_\U(M_L, x, t, n)$ carries over directly to $A_L(x, n)$.
\end{proof}
}

\smallskip
A converse of Fact~\ref{fact:hardest} holds
for {\em paddable} \cclass{NP}{}-complete languages.
(In 1977, Berman and Hartmanis observed that all known natural \cclass{NP}{}-complete languages
are paddable and conjectured that all \cclass{NP}{}-complete languages are~\cite{berman1977isomorphisms}.
The conjecture is still open.
Formally, language $L$ is {\em paddable} if there are polynomial-time
computable functions $\mbox{pad}(x,p)$ and $\mbox{extractPad}(x)$ such that,
for all strings $x$ and $p$, (1) $\mbox{pad}(x,p) \in L$ iff $x \in L$,
and (2) $p = \mbox{extractPad}(\mbox{pad}(x,p))$.)

Here is the partial converse of Fact~\ref{fact:hardest}.
The reader may safely skip the proof,
which is a standard exercise given only for completeness.
\begin{fact}\label{fact:converse}\it
  For every paddable \cclass{NP}{}-complete language $L$,
  there is a verifier $V$ that is as hard to approximate as $V_\U$;
  that is,
  if $V$ has a $d(n)$-Hamming-approximation algorithm
  (working with probability $p(n)$)
  then $V_\U$ has a $d(n)$-Hamming-approximation algorithm
  (working with probability $p(n)$).
\end{fact}
\begin{proof}  
  Fix $L$.
  By the Berman-Hartmanis isomorphism theorem~\cite[p.312]{berman1977isomorphisms},
  since both $L$ and $\U$ are paddable,
  there exists a polynomial-time isomorphism $\psi$ from strings over the alphabet of $\U$
  to strings over the alphabet of $L$, that is, a bijection such that 
  $I \in \U$ iff $\psi(I)\in L$,
  where $\psi$ and its inverse are computable in polynomial time.

  Use $\psi$ and verifier $V_\U$ to create the verifier $V$ for $L$ as follows:
  {\em $V(x,w)$ accepts iff $V_\U(\psi^{-1}(x), w)$ accepts.}
  This $V$ is a polynomial-time verifier for $L$,
  because $x\in L$ iff $\psi^{-1}(x)\in \U$.
  % Verifier $V_\psi$ runs in polynomial time
  % because $V_U$ does and computing $\psi^{-1}$ takes polynomial time.
  % To see that $V_\psi$ is a verifier for $L$, fix any $x$.  
  % By definition of $\psi$ and $V_\psi$, for $I = \psi^{-1}(x)$,
  % \[
  % x\in L
  % ~\Longleftrightarrow~
  % I \in \U
  % ~\Longleftrightarrow~
  % \exists w.~V_\U(I,w) \mbox{ accepts}
  % ~\Longleftrightarrow~
  % \exists w.~V_\psi(x,w) \mbox{ accepts}.\]

  Now, let $A$ be any approximation algorithm for verifier $V$.
  Use $A$ to create the approximation algorithm $A_\U$ for $V_\U$:
  {\em $A_\U(I,n)$ simply returns $A(\psi(I), n)$}.
  It is easy to verify that $A_\U$ runs in polynomial time,
  and that (for any $I\in\U$ and appropriate $n$) the string $a$ returned by $A_\U(I,n)$
  approximates a witness $w$ such that $V(\psi(I),w)$ accepts,
  and that (by the choice of $V$) this $w$ is also accepted by $V_\U(I,w)$.
  Thus, $A_\U$ approximates $V_\U$ 
  just as well as $A$ approximates $V$.
  % Algorithm $A_\U$ runs in polynomial time because $A$ does and $\psi$ 
  % is computable in polynomial time.  
  % Consider running $A_\U$ on any input $(I,n)$
  % for which there is some $n$-bit witness $w$ that $V_\U(I,w)$ would accept.
  % In line 1 of $A_\U$, by definition of $\psi$, then, $x = \psi(I)$ is in $L$.
  % By inspection, $w$ would also be accepted by $V(x,w)$.
  % So, in line 2 of $A_\U$
  % (since $A$ approximates $V$ as assumed in the Fact),
  % with probability at least $p(n)$, 
  % the output $w'$ of the call $A(x,n)$ 
  % is an $n$-bit string that $d(n)$-approximates 
  % some $n$-bit witness $w$ that $V(x,w)$ would accept.
  % By inspection, this witness $w$ would also be accepted by $V_\U(I, w)$.
  % Thus, $A_\U$ approximates $V_\U$ as desired.
\end{proof}
With the definitions and properties related to $\U$ noted, we return to discussing new results.

\paragraph{Arbitrary verifiers are hard to approximate within $n/2+O(\sqrt{n\log n})$.} 
By Fact~\ref{fact:converse},
the hardness of approximating $V_\U$
(Thm.~\ref{theorem:main})
extends immediately
to every paddable \cclass{NP}{}-complete language $L$,
for {\em some} verifier $V$ for $L$:
\begin{corollary}\label{corollary:main}
  Fix any $\alpha>0$.
  For any paddable \cclass{NP}{}-complete language $L$,
  there exists a verifier $V$ for $L$ with the following properties.
  If \cclass{P}{}$\ne$\cclass{NP}{}, verifier $V$ has no $(n/2+ \Hlog)$-Hamming-approximation algorithm.
  If \cclass{RP}{}$\ne$\cclass{NP}{}, verifier $V$ has no $(n/2+ \Hlog)$-Hamming-approximation algorithm that works
  with probability $1-O(1/\PPlog)$.
\end{corollary}
Previous works proved a hardness threshold of $n/2-n^{2/3+\epsilon}$, strictly below $n/2$,
so in this sense Corollary~\ref{corollary:main} strengthens those results.
Corollary~\ref{corollary:main} is weaker in that previous works 
did not require paddability of $L$,
but recall that to date no \cclass{NP}{}-complete language is known to be not paddable.

\paragraph{Many natural verifiers can be approximated within $n/2$.} 
Does a hardness threshold above $n/2$ hold for more natural verifiers?
For many natural \cclass{NP}{}-complete problems the answer is no:
their natural verifiers do have $n/2$-approximation algorithms.
Observation~\ref{observation:vc} 
describes trivial $n/2$-approximation algorithms for
the decision problems for many unweighted \cclass{NP}{} optimization problems over sets,
including the decision problems 
for the unweighted variants of Set Cover, Vertex Cover, Independent Set, and Clique.
Theorem~\ref{theorem:vc} gives combinatorial 
$n/2$-approximation algorithms for 
Vertex Cover, Independent Set, and Clique as {\em optimization problems}.
(E.g., {\em given a graph, compute a minimum-weight vertex cover}.
Approximation for optimization problems
is a-priori harder than for decision problems,
as the budget for the objective function is not given.
Enumerating possible budgets doesn't reduce optimization to decision,
because successful approximation is not easily verified.)
%These are the first non-trivial positive results that we are aware of.

\paragraph{Many natural verifiers are hard to approximate within $n/2-n^{\epsilon}$.}
For many natural verifiers, 
recall that the best previous hardness results are due to Feige et al.,
who show a hardness threshold of $n/2-n^{1-\delta}$ for a fixed small $\delta>0$.
Our last set of results 
(Theorems~\ref{theorem:SAT}, \ref{theorem:other}, \ref{theorem:HAMPATH})
increases this threshold to $n/2-n^{\epsilon}$,
reducing the exponent $1-\delta$ to any constant $\epsilon>0$,
for many of the verifiers they considered.
Unlike the previous works, these proofs do not rely on error correcting codes;
they use only simple amplification arguments similar to those of Feige et al.

\paragraph{Approximating arbitrary verifiers (upper bounds).}
The hardness result for randomized approximation algorithms (Corollary~\ref{corollary:main} part (ii))
is roughly tight in the following sense.
For any verifier $V$ of any \cclass{NP}{}-complete language and any fixed $\alpha>0$,
the naive randomized algorithm (guess a random $n$-bit string) 
is an $(n/2+\Hlog)$-approximation algorithm with probability $1-O(1/\Plog)$.
(Note that the exponent here is $4\alpha$, compared to $4\alpha+1$ in the hardness result.)

Regarding deterministic approximation algorithms,
any verifier has a trivial deterministic $(n-c)$-approximation algorithm for any fixed $c>0$.
This is a weak upper bound, but
it is the best possible for any so-called {\em black-box} algorithm.
See Section~\ref{section:arbitrary}.

%%%%%%%%%%%%%%%%%%%%%%%%%%%%%%%%%%%%%%%%%%%%%%%%%%%%%%%%%%%%
% Hardness of the universal language
%%%%%%%%%%%%%%%%%%%%%%%%%%%%%%%%%%%%%%%%%%%%%%%%%%%%%%%%%%%%

\section{The universal verifier is hard to approximate within $n/2+O(\sqrt{n\log n})$}

This section proves our core result:
the hardness of $(n/2+\H)$-approximating
the natural verifier $V_\U$ for the universal \cclass{NP}{}-complete language $\U$.
The proof is based on the fact that, for any fixed $\alpha$
and $n$-bit string $a$,
the number of $n$-bit strings that are {\em not} within Hamming distance $n/2+\H$ from $a$
is at least a polynomial fraction of all $2^n$ strings.
We prove a particular form of this fact next.
The reader can safely skip the proof, which is a standard calculation.

To simplify notation define utility functions
$H(n,\alpha) = \H$
and $P(n,\alpha) = \P$.
\begin{lemma}\label{lemma:rand}
  For any constant $\alpha > 0$,
  the number of $(n-1)$-bit strings that have more than
  $n/2+H(n,\alpha)$ ones is $\Omega(2^{n}/P(n,\alpha))$.
\end{lemma}

\begin{proof}
  Assume for simplicity of notation that $n$ is even.
  (The case when $n$ is odd is similar.)
  We first prove a weaker claim:
  {\em the number of {\em $n$-bit} strings that have more than
    $n/2+H(n,\alpha)$ ones is $\Omega(2^n/P(n,\alpha))$}.
  Then we extend the proof to prove the lemma.

  Let $H=1+\lfloor H(n,\alpha)\rfloor = \Theta(\sqrt{n\log n})$.  
  Let $p_i = {n \choose n/2 +i}$.  
  The number in question in the claim equals $\sum_{i=H}^{n/2} p_i$.  
  We lower-bound the value of this sum.
  Consider the ratio between any two successive terms:
  \[\frac{p_{i+1}}{p_i}
  \,=\,
  \frac{1 - 2i/n}{1+2(i+1)/n}.
  \]
  For $i \le H +n/H = O(H)$, the ratio is at least $1-O(H/n)$,
  so $p_{(H+n/H)}$ (considering the product of the first $n/H+1$ ratios,
  for $i\in\{H,H+1,\ldots,H+n/H\}$) 
  is at least $p_H (1-O(H/n))^{O(n/H)} = \Omega(p_H)$.
  Thus, the first $n/H$ terms in the sum total $\Omega((n/H) p_H)$.
  A calculation (details below) shows that $p_H = \Omega(2^n/n^{4\alpha+1/2})$.
  Thus, the sum is $\Omega((n/H) 2^n/n^{4\alpha+1/2})$.
  Plugging in the definitions of $H$ and $P$ and simplifying,
  this value is $\Omega(2^{n}/P(n,\alpha))$, proving the claim.

  To prove the lemma, first observe that, following the above calculations,
  the number of $n$-bit strings with at least $n/2+H+1$ ones
  (i.e.\ $\sum_{i=H+1}^{n/2} p_i$) is also $\Omega((n/H) p_H) = \Omega(2^{n}/P(n,\alpha))$.
  Deleting the first bit from any such string $w$ yields an $(n-1)$-bit string $w'$
  that must have at least $n/2+H$.
  Each such string $w'$ is obtained from at most two strings $w$ in this manner.
  Thus, the number of $(n-1)$-bit strings with $n/2+H$ or more ones
  is at least half the number of $n$-bit strings with $n/2+H+1$ or more ones.
  Thus, both numbers are $\Omega(2^{n}/P(n,\alpha))$, proving the lemma.

  Here are the promised details of the calculation that $p_H = \Omega(n^{-4\alpha-1/2})$:
  \[
  \frac{2^n}{p_H\sqrt{n}} ~=~
  \frac{2^n}{{n \choose n/2+H}\sqrt n}
  ~ = ~
  O\left(\left(1-2H/n\right)^{n/2-H}
    \left(1+2H/n\right)^{n/2+H}\right)
  ~ = ~  O\left(\exp(4H^2/n)\right)
  ~ = ~  O\left(n^{4\alpha}\right).
  \]
  The second step uses Stirling's approximation, $k! = \Theta((k/e)^k\sqrt{k})$,
  to simplify ${n\choose n/2+H}$, followed by a careful simplification of terms.
  The third step uses $(1+a)^b \le \exp(ab)$ when $|a|<1$.
\end{proof}
(We remark without proof that the bounds in Lemma~\ref{lemma:rand} are tight
up to constant factors:
a random $n-1$-bit string has more than $n/2+H(n,\alpha)$ ones
with probability $\Theta(1/P(n,\alpha))$.

\medskip
Here is the core theorem.
Recall
$H(n,\alpha) = \H$,
$P(n,\alpha) = \P$,
and $V_\U$ is the verifier for the universal \cclass{NP}{}-compete language $\U$.

\begin{theorem}\label{theorem:main} Fix constant $\alpha > 0$.

  \smallskip

  \noindent
  {\em (i)} If there is an $(n/2+H(n,\alpha))$-Hamming-approximation algorithm $A_\U$
  for $V_\U$, then \cclass{P}{}=\cclass{NP}{}.

  \smallskip
  \noindent
  {\em (ii)} If there is an $(n/2+H(n,\alpha))$-Hamming-approximation algorithm $A_\U$
  for $V_\U$ that works with probability $1-O(1/(n\, P(n,\alpha)))$, then \cclass{RP}{}=\cclass{NP}{}.
\end{theorem}

\renewcommand{\int}{\text{int}}

\begin{proof}%(of Lemma~\ref{lemma:main})
  (i) Assume that there exists $A_\U$ as in the theorem.  We show that $A_\U$
  can be used to decide $\U$ in polynomial time.
  Since $\U$ is \cclass{NP}{}-complete, it follows that \cclass{P}{}=\cclass{NP}{}.

  The basic idea is the following.  Given $(V,x,1^t)$, run $A_\U$ to find a string $a$ such that,
  if any witness exists, then some witness must be ``close'' to $a$.
  Eliminate from the universe of candidates all strings that are too far from $a$.
  By Lemma~\ref{lemma:rand}, this is a polynomial fraction of the universe.
  Repeat, each time eliminating a polynomial fraction of the remaining candidates.
  After polynomially many iterations, the universe of possible candidates
  will have polynomial size.  Use the verifier on each one to check if it's a witness.

  Of course, calling $A$ with the {\em same} input repeatedly
  would not eliminate additional candidates.  Instead we modify the verifier each iteration.
  Here are the details.
  Let $(V,x,1^t)$ be the given possible member of $\U$.
  As the algorithm iterates, it will modify $V$ and $t$,
  and it will keep an additional variable $u$,
  which is the current candidate universe size (initially $2^n$).
  Given the current universe size $u$, the corresponding universe of candidates
  will be the set containing the $u$ lexicographically smallest $n_u$-bit strings,
  where $n_u$ is $\lceil\log_2 u \rceil$.
  Abusing notation, let $[u]$ denote this set.
  (Initially $[u]$ is $[2^n]$, that is, all $n$-bit strings.)
  Given any $n_u$-bit string $a$,
  define $\N(a)$, the {\em neighborhood} of $a$, to contain
  all $n_u$-bit strings whose Hamming distance to $a$ is at most $n_u/2 + H(n_u, \alpha)$.
  (This may include some strings not in the universe $[u]$.)
  The algorithm to decide membership of $(V,x,1^t)$ in $\U$ 
  is in Fig.~\ref{fig:alg}.
  
  \medskip

  \newcommand{\comment}[1]{\hfill {\em\small #1}}

  \begin{figure}[th]
    {\bf alg}$_\U(\mbox{verifier $V$, string $x$, string $1^t$})$:
    \comment{Accepts iff $(V,x,1^t)\in \U$ (that is, iff $\exists w.~(V(x,w)$ accepts in $t$ steps$)$.)}

    \framebox{\parbox{0.99\textwidth}{

        1. ~ For each $n=0,1,2,\ldots,t$: 
        \comment{Try each possible witness size $n$.}
        
        1.1 ~~~ Define verifier $W_n(x,w)$ to do the following: 
        \comment{Modify $V$ for precondition in line 3.}
        
        \smallskip

        \centerline{\sf ``Accept iff $w\in[2^n]$ and $V(x,w)$ accepts within $t$ steps.''}
        \smallskip

        1.2 ~~~  Take $T$ polynomially larger than $t$, such that $W_n$ necessarily halts within $T$ steps.

        1.3 ~~~ Run check$_\U(n,2^n,(W_n,x,1^{T}))$.  \comment{check$_\U$ is defined below.}

        2. ~ Accept $(V,x,1^t)$ if any call to check$_\U$ returned True, else reject.
      }}

    \medskip

    {\bf check}$_\U(\mbox{integer } n, \mbox{integer } u,(V, x, 1^t))$:
    \comment{If precondition met, returns whether $\exists w.~V(x,w)$ accepts.}

    \framebox{\parbox{0.99\textwidth}{

        3. {\bf precondition:} {\em $V$ halts within $t$ steps on any input, 
      and $\{w ~|~V(x,w) \mbox{ accepts}\}$ is a subset of $[u]$.}

        4.  ~ {\bf if} $u$ is bounded by a polynomial in $n$: \comment{Base case.}

        4.1 ~~~ Return True if for some $w\in[u]$, $V(x,w)$ accepts.  Otherwise return False.

        5.  ~ {\bf else}:

        6.1 ~~~ Let $a \leftarrow A_\U((V,x,1^t),n_u = \lceil \log_2 u\rceil)$. 
        \comment{Compute approximate witness, if any witness exists.}

        \smallskip

        6.2 ~~~ Let $u' \leftarrow |[u]\cap \N(a)|$.
        \comment{Restrict to smaller universe $[u]\cap\N(a)$.}

        \smallskip

        6.3 ~~~ Define bijection $\phi$ to map the $i$th element
        of $[u']$ to the $i$th element of $\N(a)\cap[u]$.

        \smallskip

        6.4 ~~~ Define verifier $V_{a}(x,w')$ to do the following: 
        {\sf ``Accept iff $w'\in[u']$ and $V(x,\phi(w'))$ accepts.''}

        \smallskip
        
        6.5 ~~~ Take $t'$ polynomially larger than $t$ (details in proof) such that $V_{a}$ necessarily halts within $t'$ steps.

        \smallskip

        6.6 ~~~ Return check$_\U(n,u',(V_{a}, x, 1^{t'}))$.  \comment{Recurse on universe $[u']$.}
      }
    }
    \caption{Definition of algorithm alg$_\U$, which uses $A_\U$ to decide membership in $\U$.}
    \label{fig:alg}
  \end{figure}

  \noindent{\em Correctness.}
  First we show that, assuming the precondition in line 3 is met,
  check$_\U(n,u,(V,x,1^t))$ is correct.
  That is, check$_\U$ returns True
  iff there is a string $w$ such that $V(x,w)$ accepts.

  We call such a string $w$ a {\em witness for $V$}.
  In contrast, by a {\em witness for $V_\U$}, we mean a string $w$
  such that $V_\U((V,x,1^t), w)$ accepts,
  which by definition is a $w$ such that $V(x,w)$ accepts {\em within $t$ steps}.
  But note that, by the precondition on $t$ in line 3,
  $w$ is a witness for $V$ if and only if $w$ is a witness for $V_\U$.

  By inspection, if check$_\U$ returns on line 4.1, it is correct.
  Otherwise, lines 6.1-6.6 are executed, so

  \smallskip

  \noindent
  \begin{tabular}{l@{ ~ }l@{ ~ }l@{\hspace{2em}}ll}
    \multicolumn{3}{l}{$V$ has a witness}
    \\
    &$\Leftrightarrow$&  $V$ has an $n_u$-bit witness
    &\em\small
    &\em\small By precondition on $[u]$.
    \\
    &$\Leftrightarrow$&  $V_\U$ has an $n_u$-bit witness
    &\em\small
    &\em\small By precondition on $t$.
    \\
    &$\Leftrightarrow$& $V_\U$ has a witness $w$ in $\N(a)$.
    &\em\small(Line 6.1)
    &\em\small By def.\ of $A_\U$.
    \\
    &$\Leftrightarrow$& $V$ has a witness $w$ in $\N(a)$.
    &\em\small 
    &\em\small By precondition on $t$. 
    \\
    &$\Leftrightarrow$&  $V$ has a witness $w$ in $[u]\cap\N(a)$
    &\em\small
    &\em\small By precondition on $[u]$.
    \\
    &$\Leftrightarrow$&  $V_{a}$ has a witness $w'$
    &\em\small(Line 6.4)
    &\em\small By def.\ of $V_{a}$.  Note $w'=\phi^{-1}(w)$.
    \\
    &$\Leftrightarrow$& check$_\U(n,u',(V_{a},x,1^{t'}))$ accepts
    &\em\small(Line 6.6)
    &\em\small By induction (as precondition is met by $(V_{a},x,1^{t'})$).
    \\
    &$\Leftrightarrow$& check$_\U(n,u,(V,x,1^t))$ accepts.
    &\em\small(Line 6.6)
    &\em\small By inspection.
  \end{tabular}
  \medskip

  \noindent
  Thus, check$_\U$ is correct.
  Correctness of the algorithm alg$_\U$ follows
  just by inspecting alg$_\U$.
  Namely
  there is a witness such that $V(x,w)$ accepts within $t$ steps
  iff there is such a witness with $n$ bits for some $n\le t$.
  Further, $V(x,w)$ accepts some $n$-bit witness $w$ within $t$ steps 
  iff $W_n(x,w)$ accepts some $w$.
  Finally, the input $(n,2^n,(W_n,x,1^T))$ meets the precondition for the call to check$_\U$,
  so that call correctly determines whether there is a $w$ such that $W_n(x,w)$ accepts.

  \medskip
  \noindent{\em Running time.}
  To disambiguate, let $(V_0, x, 1^{t_0})$ denote the original input to alg$_\U$.
  Below, by ``polynomial'', we mean polynomial in the length of this input.

  We first show that, for any $n$, when alg$_\U$ calls
  check$_\U(n,2^n,(W_n,x,1^T))$, it results in only polynomially many recursive calls.
  Consider an execution of lines 6.1-6.6 of check$_\U$.
  Let $d=n/2+H(n_u,\alpha)$.  By Lemma~\ref{lemma:rand},
  the number of strings in $[2^{n_u-1}]$ 
  that have distance more than $d$ from the string $0^{n_u-1}$
  is $\Omega(2^{n_u} /P(n_u,\alpha))$.
  By symmetry, the number of strings in $[2^{n_u-1}]$ 
  have distance more than $d$ from the last $n_u-1$ bits of $a$
  is $\Omega(2^{n_u} /P(n_u,\alpha))$
  (using here that $[2^{n_u-1}]$ contains exactly all $(n_u-1)$-bit strings).
  By definition of $n_u$, we have $u > 2^{n_u-1}$, so, for each such
  string $z$, the $n_u$-bit string $0z$ is in $[u]$ and has distance at least $d$ from $a$.
  Thus, $\Omega(2^{n_u}/P(n_u,\alpha))$ strings in $[u]$
  have distance more than $d$ from $a$,
  so $u - u' = \Omega(2^{n_u}/ P(n_u,\alpha))$.
  This implies that check$_\U$ recurses $O(P(n_u,\alpha))) = O(P(n,\alpha))$ 
  times before $n_u$ decreases by 1.
  Since $1\le n\le n_u$ throughout, this in turn implies that check$_\U$
  recurses $O(n\,P(n,\alpha))$ times total.  

  Next consider how the encoding size of the arguments $(n,u,(V,x,1^t))$ grows as the computation progresses.
  By inspection of alg$_\U$ each triple $(W_n,x,1^T)$ has polynomial encoding size.
  With each recursive call made by check$_\U$,
  the size of the encoding of the arguments remains polynomial, 
  because $u$ and $n_u$ do not increase and, in line 6.4, the encoding of the verifier $V_{a}$ 
  is only an additive polynomial larger than the encoding of $V$
  (the verifier $V_{a}$ has encoded in it $V$, along with $u'$ and $a$, 
  which each have $n_u\le n\le t_0$ bits).
  We verify below that, when the verifier $V_{a}$ runs,
  the check ``$w'\in[u']$'' and the computation of $\phi(w')$ from $w'$ can be done in polynomial time,
  so that $t'$ is only an additive polynomial larger than $t$.
  Thus, with each recursive call, the encoding size of the arguments
  grows by an additive polynomial.  Since there are polynomially many calls,
  the encoding size remains polynomial.

  Next we consider the running time of each call check$_\U(n,u,(V,x,1^t))$.
  First consider the case that lines 6.1-6.6 execute.
  Line 6.1 executes in polynomial time by the assumption on $A_\U$
  (using that the encodings of $V$ and $t$ are polynomial, as shown above).
  Implement line 6.2 in polynomial time as follows.
  Let $N(s)$ denote the number of $n_u$-bit binary strings having string $s$
  as a prefix and having Hamming distance at most $d=\lfloor n_u/2+H(n_u,\alpha)\rfloor$
  to $a$.  The first $|s|$ bits of such a string agree with $s$; 
  the remaining $n_u-|s|$ bits differ in up to $d-\ell$ places from $a$,
  where $\ell$ is the Hamming distance between $s$ and the first $|s|$ bits of $a$.
  Thus, $N(s)=\sum_{j=0}^{d-\ell}{n_u-|s|\choose j}$.
  Thus, given $s$, $N(s)$ can be computed in polynomial time.
   Now, compute $u'$ in line 6.2 in polynomial time via the identity
  \[u' =\sum_{\ell : b_{\ell}=1} N(b_1 b_2\cdots b_{\ell-1}0),\]
  where $b$ is the binary representation of $u$ 
  (with $b_1=1$ being the most significant bit).
  (Each string $w$ in $\N(a)\cap [u]$ is lexicographically less than $b$,
  so is counted by the term for $\ell$ in the sum where $\ell = \min\{i~|~b_i = 1, w_i = 0\}$.)
  Each term is computable in polynomial time as described in the previous paragraph, 
  so the sum is computable in polynomial time.

  Likewise, in line 6.4,
  the verifier $V_{a}$ can compute $w = \phi(w')$ given any $w'\in[u']$ in polynomial time as follows.
  Abusing notation for just this paragraph, for any binary string $z$,
  let $[z]$ denote the set of $|z|$-bit binary strings 
  that are lexicographically no larger than $z$.
  Compute $i=|[w']|$, the (lexicographic) rank of $w'$ in $[u']$
  (to do so, use the fact that $w'$ is the binary representation of $i-1$).
  Then the $w$ we want (by definition of $\phi$) is the one with rank $i$ within $\N(a)\cap[u]$.
  In other words, $w$ is the string such that $|\N(a)\cap[w]| = i$.
  For any string $z$, using the previous paragraph, we can compute $|\N(a)\cap[z]|$ 
  in polynomial time, so we can use binary search over $z\in [u]$
  (checking $|\N(a)\cap[z]| \le i$) to find $w$ in polynomial time.
  
  Finally, when the base case is reached (line 4.1 is executed)
  the verifier $V$ is actually run on (polynomially many) inputs $(x,w)$.
  Since $V$ halts within $t$ steps, and (as argued previously) $t$ is polynomial,
  this step takes polynomial time.  This completes the proof for part (i).  

  \smallskip

  \noindent {\em Proof of part (ii).}
  Assume that $A_\U$ exists as in part (ii) of the theorem statement.
  Consider using that $A_\U$ directly in the algorithm alg$_\U$ described in the proof for part (i).
  First consider the case that the resulting algorithm is given a positive instance of $\U$,
  that is, an instance that has some witness $w$.
  Let $n$ be the length of the witness,
  and consider the corresponding call to check$_\U$ by alg$_\U$.
  As observed in analyzing the number of iterations,
  for each $n_u\in\{\lfloor c\log n\rfloor, \ldots, n\}$,
  this call results in $O(P(n_u,\alpha))$ recursive calls to check$_\U$ with that $n_u$.
  Each recursive call calls $A_\U$ once.
  Since $A_\U$ is randomized and (by assumption)
  has probability $O(1/n_u P(n_u,\alpha))$ of failure on any input $((V,x,1^t),n_u)$,
  the probability that none of these calls to $A_\U$ fails is at least
  \[
  \prod_{n'=c\log n}^n \Big(1-O\Big(\frac{1}{n' P(n',\alpha)}\Big)\Big)^{O(1/P(n',\alpha))}
  \,\le\,
  \exp\Big(-O\Big(\sum_{n'=1}^n \frac{1}{n'}\Big)\Big)
  \,=\,
  \exp(-O(\log n))
  \,=\,
  \frac{1}{n^{O(1)}}.
  \]
  Thus, the algorithm has probability $1/n^{O(1)}$ of accepting the input.
  Since the algorithm never has false positives,
  this shows that $\U\in\,$\cclass{RP}{}.  Since $\U$ is \cclass{NP}{}-complete,
  the result follows.
\end{proof}

%%%%%%%%%%%%%%%%%%%%%%%%%%%%%%%%%%%%%%%%%%%%%%%%%%%%%%%%%%%%
% Algorithms for Vertex Cover and related
%%%%%%%%%%%%%%%%%%%%%%%%%%%%%%%%%%%%%%%%%%%%%%%%%%%%%%%%%%%%

\section{Many natural verifiers can be approximated within $n/2$}
\label{section:algorithms}
By Corollary~\ref{corollary:main} to Thm.~\ref{theorem:main},
every (paddable) \cclass{NP}{}-complete problem has {\em some} verifier that is hard to 
approximate within $n/2+O(\sqrt{n\log n})$.
These hard verifiers are not natural verifiers for the problems in question.
Next we observe that that natural verifiers can be easier to approximate:
for many \cclass{NP}{}-complete problems, 
there are $n/2$-approximation algorithms for a natural verifier.

As examples of natural verifiers, the natural verifier for 3-SAT 
uses witnesses that are $n$-bit strings, where the $i$th bit is 1
if the $i$th variable is assigned True.
For Vertex Cover (and other optimization problems over subsets), 
the natural verifier uses witnesses that are $n$-bit strings,
where the $i$th bit is 1 if the $i$th vertex is in the cover.
% Notationally, in this section we use $x$ or
% $x^*$ to refer to witnesses (or solutions) of the problems.

For a few problems, $n/2$-approximation algorithms follow trivially by symmetry.
For example, for any instance of not-all-equal 3-SAT (NAE-3SAT), for any satisfying 
assignment, complementing all the variables also gives a satisfying assignment.
Thus, the all-False assignment gives an $n/2$-approximation.

Next focus on decision versions of optimization problems.
For intuition, consider the decision version of unweighted Vertex Cover:
{\em Given a graph $G=(V,E)$ and an integer $k\le |V|$,
does there exist a vertex cover (a subset of vertices incident to every edge) of size at most $k$?}
Let $n=|V|$ be the number of vertices.

The natural verifier $V_{VC}$, on input $((G,k), C)$, where $C\subseteq V$,
accepts iff $|C| \le k$ and $C$ is a vertex cover.
Here is a trivial $n/2$-approximation algorithm $A_{VC}$:
{\em On input $((G=(V,E),k), n)$, if $k\le n/2$, output the empty set, otherwise output $V$.}
We assume here the natural encoding of the set $C$ as an $n$-bit string,
whose $i$th bit determines whether $C$ contains the $i$th vertex in $V$.

To see that $A_{VC}$ is an $n/2$-approximation algorithm,
suppose that $G$ does have a vertex cover $C$ of size at most $k$.
Then there also exists a vertex cover $C'$ of size exactly $k$.
If $k\le n/2$, then $|C'|\le n/2$, so $C'$ has Hamming distance
at most $n/2$ from the empty set.
Otherwise $k>n/2$ and $|C'|>n/2$, so $C'$ has Hamming distance
at most $n/2$ from $V$.

The same idea gives $n/2$-approximation algorithms for the unweighted
versions of similarly structured \cclass{NP}{}-complete problems such as Set Cover,
Independent Set, and Clique.  (The latter two are maximization problems,
but the idea is the same.)

\begin{observation}\it
\label{observation:vc}
Consider any \cclass{NP}{} verifier $V$ that, on input $((I,k),w)$,
accepts only if $w$ is the natural encoding of 
a subset $S$ of some $n$-element universe $U(I)$
such that $|S|\le k$.  

Consider any \cclass{NP}{} verifier $V$ that, on input $((I,k),w)$,
accepts only if $w$ is the natural encoding of 
a subset $S$ of some $n$-element universe $U(I)$
such that $|S|\ge k$.  
 
Any such verifier has an $(n/2)$-Hamming-approximation algorithm.
\end{observation}
(Note that ``only if'' in the observation is not ``if and only if''.)
\smallskip

Observation~\ref{observation:vc} applies to many problems 
(unweighted Vertex Cover, Independent Set, Clique, Set Cover, Feedback Vertex Set, 
and, e.g., Hamiltonian Path if the witness is encoded as an edge set),
but it is somewhat unsatisfying in that it is an artifact of the fact that the problems are
 presented as {\em decision} problems, and not in their more natural {\em optimization} form.
For example, the optimization form of unweighted Vertex Cover is,
{\em given a graph $G$, compute a Vertex Cover $C$ of minimum size}.
In this context, by a {\em $d(n)$-approximation algorithm},
we mean a polynomial-time algorithm that, given just $G$ and $n$
(but not the size of the desired cover!)
outputs a subset $S$ of the $n$ vertices, where $S$
has Hamming distance at most $d(n)$ from some
{\em minimum-size} vertex cover.

\begin{theorem}\label{theorem:vc}
  There are $n/2$-Hamming-approximation algorithms for the optimization versions of
  unweighted Vertex Cover, Independent Set, and Clique.
\end{theorem}
\begin{proof}

The algorithms for Independent Set and Clique work by standard reductions to Vertex Cover.
The algorithm for Vertex Cover is based on a classic result of Nemhauser and Trotter.
\begin{proposition}[\cite{nemhauser1975vertex,khuller2002algorithms}]\label{prop:nt}
Fix any instance $I$ of Vertex Cover.
Let $y$ be any (minimum cost) basic feasible solution to the linear program
relaxation of the standard integer linear program for $I$.
Then, for each vertex $v$, the variable $y_v$ has value in $\{0,1/2,1\}$,
and there exists an optimal vertex cover $C^*$
that has the following property.  For each vertex $v$,
if $y_v=0$, then $v\not\in C^*$,
while if $y_v = 1$, then $v\in C^*$.
\end{proposition}

The $n/2$-approximation algorithm for the optimization form of unweighted
Vertex Cover is as follows:
{\em Given the instance $G$,
compute the minimum-cost basic feasible solution $y$,
then output the set of vertices $S=\{v \in V ~|~ y_v > 0\}$.}

Since the basic feasible solution $y$ can be computed in polynomial time,
the algorithm clearly runs in polynomial time.
Next we prove the approximation guarantee.
Let $C^*$ be the optimal vertex cover from Prop.~\ref{prop:nt}.
Since there is a feasible solution of cost $|C^*|$ to the linear program,
while $y$ is a minimum-cost solution, it follows that
$|C^*| \ge \sum_i y_i$.
The choice of $S$ implies $\sum_i y_i \ge |S|/2$.
Thus, $|C^*| \ge |S|/2$.

The choice of $S$ also implies that $C^* \subseteq S$.
This and $|C^*| \ge |S|/2$ imply that $C$ 
is within Hamming distance $|S|/2$ (which is at most $n/2$)
from $C^*$.

This proves the approximation guarantee for Vertex Cover.
The other problems follow, because the standard reductions
between these problems preserve Hamming approximation.
Here are the details:

The $n/2$-approximation algorithm for Independent Set is as follows:
{\em Given the instance $G$, run the $n/2$-approximation algorithm
  for Vertex Cover.  Let $S$ be the output.  Return the complement of $S$, 
  i.e., $\overline{S} = V-S$.}
The output is an $n/2$-approximation because a vertex set
$I^*\subseteq V$ is a maximum independent set if and only 
if its complement $\overline I^* = V-I^*$ is a minimum vertex cover.

%\smallskip

The $n/2$-approximation algorithm for Clique is as follows:
{\em Given the instance $G=(V,E)$, 
  compute the complement graph $G'=(V,\overline E)$
  with the same vertex set but whose edge set is
  the complement of $E$.
  Run the $n/2$-approximation algorithm for Independent Set on $G'$.
  Let $S$ be the output.  Return $S$.}
The output is an $n/2$-approximation because a vertex set is a maximum clique in $G$ 
if and only if it is a maximum independent set in $\overline G$.
\end{proof}

%%%%%%%%%%%%%%%%%%%%%%%%%%%%%%%%%%%%%%%%%%%%%%%%%%%%%%%%%%%%
% Improving hardness results of Feige et al.
%%%%%%%%%%%%%%%%%%%%%%%%%%%%%%%%%%%%%%%%%%%%%%%%%%%%%%%%%%%%

\section{Many natural verifiers are hard to approximate within $n/2-n^\epsilon$.}

This section describes how to strengthen many of the hardness results of
Feige et al., to show that, for many of the natural verifiers that they consider,
(for any $\epsilon>0$, assuming \cclass{P}{}$\ne$\cclass{NP}{})
there are no $(n/2-n^\epsilon)$-approximation algorithms.
The proofs here are elementary padding arguments,
similar in spirit to the those of Feige et al.

The proofs use the following standard observation:
\begin{observation}\label{observation:single}\it
  {\em (i)} If the following problem has a polynomial-time algorithm, then \cclass{P}{}=\cclass{NP}{}:
  {\sf Given a 3-SAT formula, find a {\em feasible value} for the first
    variable in the formula, if one exists.}
  \smallskip

  \noindent
  {\em (ii)} If there is a randomized polynomial-time algorithm that 
  solves any $n$-variable instance of the above problem 
  with probability $1/2+1/n^c$ for any fixed $c>0$, then \cclass{RP}{}=\cclass{NP}{}.
\end{observation}
If the formula is satisfiable, a {\em feasible} value for the variable 
is one that is consistent with some satisfying assignment.
(If the formula is not satisfiable, any value can be found.)

Part (i) holds by standard arguments.
To see why part (ii) is true, note that the randomized algorithm
could be used to find a satisfying assignment of a given formula
with high probability: to determine the likely value of the first variable,
run the randomized algorithm, say, $n^{c+2}$ times, then take the majority
value --- this standard amplification trick boosts the probability 
of finding a feasible value to at least $1-1/n^2$.  
Then substitute the likely feasible value for the variable, simplify the formula,
then recurse.  This would find a full satisfying assignment with probability
at least $1-O(1/n)$ in polynomial time, showing that \cclass{RP}{}=\cclass{NP}{}.

\smallskip
We start by showing hardness of approximating the natural 3-SAT verifier:
\begin{theorem}\label{theorem:SAT}
  Fix any constants $\epsilon,c>0$.

  \noindent
  {\em (i)}
  If the natural verifier for 3-SAT
  has an $(n/2-n^\epsilon)$-Hamming-approximation algorithm $A$,
  then \cclass{P}{}=\cclass{NP}{}.
  
  \smallskip
  
  \noindent
  {\em (ii)} If that verifier
  has a randomized $(n/2-n^\epsilon)$-Hamming-approximation algorithm $A$ working with
  probability $1/2+1/n^{c}$, then \cclass{RP}{}=\cclass{NP}{}.
\end{theorem}
\begin{proof}
  {\em (i)} The proof is an elementary amplification argument.  

  Assume without loss of generality that $1/\epsilon$ is an integer 
  (otherwise replace $\epsilon$ by $1/\lceil 1/\epsilon\rceil$).

  Assume the algorithm $A$ in the statement of Theorem~\ref{theorem:SAT} exists.
  Given any 3-SAT formula $\psi$, we compute, in polynomial time,
  a feasible value for the first variable $z$
  as follows.

  To $\psi$, add $N=n^{1/\epsilon}$ copies of $z$.
  Specifically, add new clauses
  $(z^1=z) \wedge (z^2=z) \wedge \cdots \wedge (z^{N}=z)$
  (where $a=b$ is shorthand for $(a \vee \overline{b}) \wedge (\overline{a} \vee  b)$,
  and $z^1, z^2,\ldots, z^{N}$ are new variables).
  This gives a formula $\psi'$ with $n' = n + n^{1/\epsilon}$ variables,
  essentially preserving any satisfying assignments,
  but forcing the $n^{1/\epsilon}$ added variables to take the same value as $z$
  in any satisfying assignment.

  Run the approximation algorithm $A$ on $\psi'$, and let $a'$ be the returned value.
  If $\psi'$ is satisfiable, then $a'$ achieves Hamming distance
  at most $n'/2 - (n')^{\epsilon}$ to an assignment $w'$ satisfying $\psi'$.
  That is, $a'$ agrees with $w'$ on at least
  $n'/2 + (n')^{\epsilon} > n^{1/\epsilon}/2 + n$ variables.   
  To do so, even if $x$ agrees with $w'$ on {\em all} $n$ of the original variables,
  it would still have to agree with $w'$ on more than {\em half} ($n^{1/\epsilon}/2)$ of the duplicates 
  of $z$.
  Thus, the majority value of the duplicate variables in $a'$
  (true or false, whichever $a'$ assigns to more duplicates)
  must be the value that $w'$ assigns to $z$.
  This value must also be a feasible value for $z$ in $\psi$ (if one exists).

  Thus, if $A$ exists, then one can compute a feasible value for $z$ in $\psi$ (if one exists)
  in polynomial time.
  By Observation~\ref{observation:single}, then, \cclass{P}{}=\cclass{NP}{}.

  \medskip

  \noindent
  {\em (ii)}
  Assume the algorithm $A$ exists, and use it as in the proof for part (i).
  Given the formula $\psi$ with $n$ variables, 
  call $A$ on the formula $\psi'$ with $n'$ variables,
  where $n'=n+n^{1/\epsilon}$.
  By the properties of $A$ assumed in the theorem,
  the probability that the call succeeds in finding a feasible value 
  for $z$ (if one exists) is $1/2+1/(n')^{c} \ge 1/2+1/n^{O(c/\epsilon)}$.
  Thus, by Observation~\ref{observation:single} (b), \cclass{RP}{} would equal \cclass{NP}{}.
\end{proof}

Next we sketch how the same idea applies to other problems.
\begin{theorem}\label{theorem:other}  Fix any constant $\epsilon>0$.

  \noindent
  {\em (i)} If the natural verifier for unweighted Vertex Cover, Independent Set, or Clique
  has
  an $(n/2-n^{\epsilon})$-Hamming-approximation algorithm, then \cclass{P}{}=\cclass{NP}{}.

  \medskip

  \noindent
  {\em (ii)} If any of these verifiers
  has
  a randomized $(n/2-n^{\epsilon})$-Hamming-approximation algorithm
  working with probability $1/2+1/n^c$, then \cclass{RP}{}=\cclass{NP}{}.
\end{theorem}
\begin{proof}
  (i) We sketch the proof for Vertex Cover.
  The rest follow via standard reductions,
  from Vertex Cover to Independent Set
  and from Independent Set to Clique,
  as these reductions preserve Hamming approximation.

  Suppose such an algorithm $A$ exists.
  Given a graph $G=(V,E)$ with $n$ vertices,
  a non-isolated vertex $v\in V$,
  and the minimum size, $k$, of any vertex cover in $G$,
  we will use $A$ to determine in polynomial time
  either (a) that $G$ has a size-$k$ vertex cover containing $v$,
  or (b) that $G$ has a size-$k$ vertex cover not containing $v$.
  By standard arguments, if this can be done in polynomial time, then \cclass{P}{}=\cclass{NP}{}.

  Determine (a) or (b) as follows.
  Construct graph $G'$ from $G$ by adding 
  a copy $v'$ of $v$ (with edges to all neighbors of $v$),
  and a path $P$ (of new vertices)
  connecting $v$ to $v'$,
  so that $|P|$ (the number of edges in $P$)
  is even and roughly equals $n^{1/\epsilon}$.
  Let $n'=n+|P|$ be the number of vertices in $G'$,
  and let $k'=k + |P|/2$.
  Denote the successive vertices on path $P$ as $v = v_0, v_1, v_2, \ldots, v_{|P|}=v'$.
  Let $P_0$ contain the ``even'' vertices $v_2, \ldots, v_{|P|}=v'$ (not including $v$).
  Let $P_1$ contain the ``odd'' vertices $v_1, v_3, \ldots, v_{|P|-1}$.

  Run $A$ on the instance $\langle G', k'\rangle$
  and let $C_A$ (interpreted as a vertex set) be the output.
  Define the Hamming distance of $C_A$ to $P_i$  ($i\in\{0,1\}$)
  to be the number of vertices $v\in P$ that are in exactly one of the two sets $C_A$, $P_i$.
  Return {\em(a) $G$ has a size-$k$ vertex cover containing $v$}
  if the Hamming distance from $C_A$ to $P_0$
  is less than the Hamming distance from $C_A$ to $P_1$.
  Otherwise, return {\em (b) $G$ has a size-$k$ vertex cover not containing $v$.}
  
  To finish the proof we sketch why this procedure is correct.
  Let $C'$ be any minimum-size vertex cover of $G'$.
  By standard arguments, $C'$ has size $k'$ and one of two cases holds:
  \begin{description}
  \item[Case (1)] $C' = C\cup P_0$ where $C$ is a size-$k$ vertex cover of $G$ and $v\in C$, or
  \item[Case (2)] $C' = C\cup P_1$ where $C$ is a size-$k$ vertex cover of $G$ and $v\not\in C$.
  \end{description}
  By assumption, $C_A$ has Hamming distance at most $n'/2 - (n')^{\epsilon}$
  to some such $C'$.
  That is, $C_A$ agrees with $C'$ on at least 
  $n'/2 + (n')^{\epsilon} > n^{1/\epsilon}/2 + n$ vertices.
  Thus, focusing just on vertices in $P$,
  $C_A$ agrees with $C'$ on more than half of the vertices in $P$.
  
  If Case (1) above occurs for $C'$,
  then the Hamming distance from $C'$ to $P_0$ must be less than $|P|/2$,
  so the Hamming distance from $C'$ to $P_1$ must be more than $|P|/2$,
  so the algorithm returns 
  {\em (a) $G$ has a size-$k$ vertex cover containing $v$}.
  This is correct, given that Case (1) occurs.

  By a similar argument, if Case (2) occurs, then the algorithm returns
  {\em (b) $G$ has a size-$k$ vertex cover not containing $v$},
  which is correct in this case.

  This proves part (i).  The proof for part (ii) follows
  just as part (ii) of Thm.~\ref{theorem:SAT} follows
  from part (i) in the proof of Thm.~\ref{theorem:SAT}.
\end{proof}

\begin{theorem}\label{theorem:HAMPATH}
  Fix any constants $\epsilon,c> 0$.

  \noindent{\em (i)}
  Suppose that, for unweighted, directed Hamiltonian Cycle,
  the verifier that uses edge-subsets for witnesses
  has an $(n/2-n^\epsilon)$-Hamming-approximation algorithm.  Then \cclass{P}{}=\cclass{NP}{}.
  \smallskip

  \noindent{\em (ii)}
  Suppose that verifier
  has a randomized $(n/2-n^\epsilon)$-Hamming-approximation algorithm
  that works with probability $1/2+1/n^c$.  Then \cclass{RP}{}=\cclass{NP}{}.
\end{theorem}
\begin{proof}
  {\em (i)}
  Consider any polynomial-time reduction
from 3-SAT to unweighted, directed Hamiltonian Cycle.
By definition, such a reduction works as follows.
Given a 3-SAT formula $\psi$,
the reduction produces, in polynomial time,
a directed graph $G=(V,E)$ such that $G$ has a Hamiltonian cycle
if and only if $\psi$ is satisfiable;
further, given any Hamiltonian cycle $C$ in $G$,
the reduction describes how to
compute an assignment $A(C)$ satisfying $\psi$
in polynomial time.

There exist well known reductions
(e.g., see~\cite{sipser2006introduction})
such that $\psi$, $G$, and $A(\cdot)$ 
have the following further properties.
For any variable $z$ in $\psi$,
there are a pair of edges $(u,v)$ and $(v,u)$
such that, for any Hamiltonian cycle $C$,
either $C$ contains $(u,v)$ and $A(C)$
assigns $z =$ True,
or $C$ contains $(v,u)$ and $A(C)$
assigns $z =$ False.

Assume the algorithm $A$ from the observation exists.
We describe below how to modify any reduction with the above properties
so as to solve the following problem in polynomial time:
{\em Given a 3-SAT formula $\psi$, determine a feasible
value for the first variable in $\psi$ (if any exists}).
By Observation~\ref{observation:single},
this is enough to prove \cclass{P}{}=\cclass{NP}{}.

Given $\psi$, apply the reduction with the above properties
to compute the graph $G=(V,E)$.  Then, for the first variable,
say, $z$ in $\psi$, let $(v,v')$ and $(v',v)$ be the two edges
in $G$ for $z$ as described above.
Replace the edges $(v,v')$ and $(v,v')$, respectively,
with paths $P_0 = (v=v_0,v_1,v_2,\ldots,v_{k},v_{k+1} = v')$
and $P_1 = (v,v_{k},v_{k-1},\ldots,v_1,v')$,
where $v_1,v_2,\ldots,v_{k}$ are new vertices
and $k=n^{1/\epsilon}/2$, where $n=|E|$ is the number of edges in $G$.
Say that each edge $(v_i,v_{i+1})$ is a {\em duplicate} of $(v,v')$,
and that each edge $(v_{i+1}, v_i)$ is a duplicate of $(v',v)$.
Let $P=P_0\cup P_1$ be the set of all duplicate edges.
Call the resulting graph $G'$, and let $n' = n + n^{1/\epsilon}$
be the number of edges in $G'$.

Run the algorithm $A$ on $G'$, and let $P'$ be the output
(interpreted as an edge set).
Define the Hamming distance from $P'$ to $P_i$ (for $i=0,1$)
to be the number of edges $e$ in $P$
that are in exactly one of the two sets $P'$, $P_i$.
If the Hamming distance from $P'$ to $P_0$
is less than the Hamming distance from $P'$ to $P_1$,
then return {\em (a) The value True is feasible for the variable $z$}
and otherwise return {\em (b) The value False is feasible for the variable $z$}.

\smallskip 
To finish, we prove that this procedure determines a feasible value for $z$,
if there is one.
\smallskip

Assume that there is a feasible value for $z$ (that is, that $\psi$ is satisfiable).
The output $P'$ of $A$ then has Hamming distance 
at most $n'/2 - (n')^{\epsilon}$ to some Hamiltonian cycle $C'$ in $G'$.
That is, $P'$ agrees with some $C'$ on at least
$n'/2 + (n')^{\epsilon} > n^{1/\epsilon}/2 + n$ edges.
Thus, $P'$ agrees with $C'$ on strictly more than
$n^{1/\epsilon}/2$ (half) of the edges in $P$.
By the properties of the reduction, one of two cases holds:
\begin{description}
\item[Case 1.] $C' \cap P = P_0$,
  and True is a feasible value for $z$.
  In this case, the Hamming distance from $P'$ to $P_0$
  must be less than $|P|/2$,
  so the Hamming distance from $P'$ to $P_1$
  must be more than $|P|/2$,
  so the procedure returns
  {\em (a) The value True is feasible for the variable $z$},
  which is correct.

\item[Case 2.] $C'\cap P = P_1$,
  and False is a feasible value for $z$.
  By similar reasoning, the procedure is correct in this case as well.
\end{description}

  This proves part (i).  The proof for part (ii) follows
  just as part (ii) of Thm.~\ref{theorem:SAT} follows
  from part (i) in the proof of Thm.~\ref{theorem:SAT}.
 \end{proof}

\section{Black-box algorithms for approximating arbitrary verifiers}\label{section:arbitrary}
By Facts~\ref{fact:hardest} and~\ref{fact:converse},
the problem of approximating arbitrary \cclass{NP}{} verifiers
is equivalent to the problem of approximating the verifier $V_\U$ for the universal \cclass{NP}{}-complete
language $\U$.

\smallskip

\noindent{\em Randomized algorithms.} 
The best randomized approximation algorithm for $V_\U$ that we know of is the trivial algorithm:
guess $n$ random bits.  By Lemma~\ref{lemma:rand}, this 
achieves $(n/2+H(n,\alpha))$-Hamming approximation with probability $1-O(1/P(n,\alpha))$.
(As observed in the introduction, this nearly matches the main hardness result here for randomized approximation algorithms for $V_\U$.)

\smallskip

\noindent{\em Deterministic algorithms.} 
The best deterministic polynomial-time approximation algorithm for $V_\U$ that we know of
is as follows: {\em Test each $n$-bit candidate string with $c$ or fewer 1's.  If one is a witness,
return it, otherwise return $1^n$.}  
This is an $(n-c)$-approximation algorithm (for any constant $c$).  
(Note that $1^n$ is within
Hamming distance $n-c$ of all untested strings, and therefore within
Hamming distance $n-c$ of any witness if the algorithm returns $1^n$.)

\smallskip

\noindent{\em Lower bound for deterministic ``black-box'' algorithms.} 
An adversary argument shows that the above algorithm is near-optimal
among deterministic algorithms that use the verifier $V_\U$ as a {\em black box}
(by which we mean that, given any possible instance $I=(V,x,1^t)$, the algorithm determines 
information about $I$ only by querying the \cclass{NP}{} verifier $V_\U(I,w)$ with this $I$ 
and various potential witnesses $w$).  The adversary behaves as follows:
{\em Whenever the algorithm queries $V_\U(I,w)$ with a given choice of $w$, 
the adversary has $V_\U$ return ``no''.}
Suppose for contradiction that the algorithm runs in $o(n^{c-1})$ time for some constant $c$,
before it returns its answer $a$.
There are at least ${n \choose n-c+1} = \Omega(n^{c-1})$ strings 
whose Hamming distance to $a$ is $n-c+1$.
At least one of these strings $w'$ was not queried by the algorithm.
The adversary can take $w'$ to be the true witness, that is, 
it can make $V$ be such that $V(x,w)$ accepts in $t$ steps iff $w=w'$.
Then, for this instance $I$, the algorithm's answer $a$ does not achieve Hamming distance $n-c$.
Thus, any ``black-box'' deterministic algorithm that achieves
Hamming distance $n-c$ must take time $\Omega(n^{c-1})$ in the worst case.

%%%%%%%%%%%%%%%%%%%%%%%%%%%%%%%%%%%%%%%%%%%%%%%%%%%%%%%%%%%%
% Discussion
%%%%%%%%%%%%%%%%%%%%%%%%%%%%%%%%%%%%%%%%%%%%%%%%%%%%%%%%%%%%

\section{Gaps between lower and upper bounds}\label{section:discussion}
For deterministic Hamming-approximation algorithms for $V_\U$,
there is still a large gap between the lower bound ($n/2+O(\sqrt{n\log n})$)
and the upper bound ($n-O(1)$).
As discussed above, the upper bound 
cannot be improved for ``black-box'' algorithms.
As for the lower bound,
improving it would also require a different proof technique,
one that does more than use the supposed approximation algorithm $A_\U$ 
as a black box.
(This is simply because (1) such a proof technique cannot distinguish
between a deterministic algorithm and a randomized algorithm 
that fails with only exponentially small probability,
and (2) a stronger lower bound {\em does not hold} for such randomized algorithms:
guessing $n$ random bits achieves Hamming distance $n/2+O(\sqrt{n\log n})$ 
with high probability.)
Similarly, it seems likely that standard derandomization methods 
might yield a deterministic \cclass{P/POLY}\ algorithm for achieving 
Hamming distance $n/2+O(\sqrt{n\log n})$ for $\U$,
in which case a stronger lower bound would have to distinguish
polynomial-time algorithms from polynomial-time algorithms with polynomial advice.

\smallskip

Regarding the hardness of approximating {\em natural} verifiers,
we know that there are $n/2$-approximation algorithms
for natural verifiers for unweighted Vertex Cover, Clique, Independent Set, and similar problems.
The gap for these problems is smaller, since the lower bound is $n/2-n^\epsilon$.
But the gap could still be reduced for these problems.
For other problems the gap is larger,
mainly because we do not yet have $n/2$-approximation algorithms.
These include, for example, the weighted versions of the above problems, and 3-SAT.
The following leading problem is still open, both for deterministic algorithms
and for randomized algorithms that work with high probability:
\begin{quote}
  {\em Given a satisfiable 3-SAT formula,
  how hard is it to find an assignment to the variables 
  that has Hamming distance at most $n/2$ to a satisfying assignment?}
\end{quote}

\section{Acknowledgments}
Thanks to two anonymous referees for their careful reading and constructive suggestions,
including Observation~\ref{observation:vc}.

\newpage

\bibliographystyle{tocplain}   %%% \bibliographystyle{plain}
%%% !!!
%%% Change this to match the name of your BIB file
\bibliography{arxiv_v2_bib}

\end{document}